\documentclass[pdflatex,sn-mathphys-ay]{sn-jnl}

\usepackage{graphicx}
\usepackage{amsmath,amssymb,amsfonts,amsthm,mathtools,bm}
\usepackage{enumitem}
\usepackage{booktabs}

\theoremstyle{thmstyleone}
\newtheorem{theorem}{Theorem}[section]
\newtheorem{proposition}[theorem]{Proposition}
\newtheorem{lemma}[theorem]{Lemma}

\theoremstyle{thmstylethree}
\newtheorem{definition}[theorem]{Definition}
\newtheorem{assumption}[theorem]{Assumption}
\newtheorem{example}[theorem]{Example}
\newtheorem{remark}[theorem]{Remark}

\newcommand{\E}{\mathbb E}
\newcommand{\Pp}{\mathbb P}
\newcommand{\F}{\mathbb F}
\newcommand{\G}{\mathbb G}
\newcommand{\D}{\mathbb D}
\newcommand{\Hh}{\mathcal H}
\newcommand{\Kk}{\mathcal K}
\newcommand{\R}{\mathbb R}
\newcommand{\Ran}{\operatorname{Ran}}

\newcommand{\doop}{\operatorname{do}}
\newcommand{\ip}[2]{\left\langle #1,#2\right\rangle}
\newcommand{\norm}[1]{\left\lVert #1\right\rVert}

\begin{document}

\title[Causal price-of-risk mandates]{Causal Price-of-Risk Mandates under Overlapping Information}

\author*[1,2]{\fnm{Alejandro} \sur{Rodriguez Dominguez}}\email{arodriguez@miraltabank.com}

\affil[1]{\orgdiv{Head of Quantitative Analysis}, \orgname{Miralta Finance Bank},
\orgaddress{\city{Madrid}, \country{Spain}}}

\affil[2]{\orgdiv{Department of Computer Science}, \orgname{University of Reading},
\orgaddress{\city{Reading}, \country{United Kingdom}}}

\abstract{We study whether causal risk mandates constructed from overlapping
information blocks can be implemented by one self-financing portfolio that is
optimal under pooled information. In an incomplete continuous Brownian market,
signal-responsive exposures are projected onto traded Brownian directions and
embedded as closed subspaces of a predictable Hilbert space. Each causal
mandate is the attainable projection of an identified, baseline-centered
intervention surface and is therefore fixed upstream of the allocator. A first
exact decomposition separates non-traded structural response, incompatibility
across local mandates, and distortion of the common traded book. A fusion
operator corrects duplication of common risk directions. The main theorem
shows that exact causal decentralization holds if and only if the pooled
enlargement preserves reference martingales, the pooled optimum has no
component outside the additive desk span, the local causal shadows are
compatible with one common book, and that book equals the additive pooled
optimum. Under immersion, log-growth regret separates into pooled-interaction
loss and common-book distortion, while raw causal implementation error also
contains non-attainability and incompatibility. No common benchmark can reduce
incompatibility. We derive the unique allocation under a quadratic
causal-mandate penalty and exact comparative statics in a two-desk model with
shared, desk-specific, and pooled-interaction risk. Additional results cover
joint interventions, representation stability, admissible coalitions, and a
covariance-weighted continuous-semimartingale formulation. The conclusions are
partial-equilibrium statements for continuous markets.}

\keywords{causal price of risk, decentralized portfolio management,
overlapping information, log-optimal portfolio, filtration immersion, fusion
frames}

\pacs[MSC 2020]{91G10, 60G44, 46C05, 62H12}

\maketitle

\section{Introduction}
\label{sec:introduction}

Investment organizations routinely divide a pooled opportunity set into
specialist mandates.  The local desks do not necessarily see disjoint risks:
they may condition on overlapping signals, use a common benchmark, and form
causal views by standardizing different interventions.  The central allocator
then faces a question that is more restrictive than ordinary portfolio
aggregation.  Do the separately specified causal mandates correspond to the
local projections of one traded exposure, and can that exposure coincide with
the portfolio that would be chosen under pooled information?

Existing work on decentralized investment management studies how a central
decision maker can combine specialist decisions or align delegated managers
\citep{Sharpe1981,EltonGruber2004,VanBinsbergenBrandtKoijen2008}, while
\citet{BlakeEtAl2013} document the institutional consequences in pension-fund
management.  Organizational models also show that decentralizing information
processing can be costly when
agents' decisions interact \citep{Vayanos2003}.  Those contributions establish
the economic importance of delegation and information aggregation.  They do
not, however, answer the narrower stochastic-finance question considered here:
whether baseline-centered causal price-of-risk shadows obtained from
overlapping predictable information sets are jointly realizable as one
self-financing book.

Four logically different requirements enter that question.  First, a raw
structural response need not point in a Brownian direction spanned by traded
assets, so causality must be separated from attainability.  Second, singleton
information blocks need not span a pooled interaction exposure.  Third,
overlap means that summing desk quantities double-counts common risk
directions.  Fourth, pooling information can change the reference martingale
innovations; an algebraically correct reconstruction is not automatically a
valid reference-market trading interpretation.  Treating these four issues as
one generic ``information loss'' conceals which defect can be repaired by an
allocator and which cannot.

We place all attainable driver-responsive exposures in one predictable
$L^2$ space and use the minimal market price of risk as the traded target.
Relative to a declared common baseline, each desk has an observational shadow
and an attainable causal shadow.  The desk subspaces form a fusion system: its
analysis operator sends a common book to its tuple of local projections, and
the inverse fusion operator reconstructs shared directions only once.  The
orthogonal complement of the additive desk span records pooled-only
interaction opportunities.

The main result gives a necessary-and-sufficient characterization of exact
causal decentralization.  It requires pooled immersion, absence of a
pooled-only optimal exposure, compatibility of the standalone causal shadows,
and equality between their reconstructed common book and the additive pooled
optimum.  The same geometry yields a hierarchy of exact error decompositions.
The
allocator's log-growth regret contains the pooled interaction and the distance
between two implementable common books.  The tuple of local mandate errors
contains, in addition, the component orthogonal to the range of the analysis
operator.  This last component is an infeasibility certificate: no common
baseline or additive traded book can satisfy all local causal mandates.
Before this range question is reached, an exact structural-to-book identity
separates the part of each identified intervention surface that is outside the
traded local space.  Raw causal implementation error therefore has three
orthogonal sources: non-attainability within desks, incompatibility across
desks, and displacement of the common traded book.  Only the last source can be
changed by choosing another additive book.

The distinction has a direct institutional consequence.  A benchmark change
can shift what all desks share but cannot alter their baseline-centered local
projections.  We prove that the smallest mandate error over every common
benchmark and additive book is exactly the incompatibility residual.  We then
derive the unique book chosen by an allocator who trades off pooled
log-growth against a declared quadratic causal-mandate penalty.  As the
penalty becomes large, the traded book converges to the least-squares causal
reconstruction, while the objective diverges whenever hard compliance is
infeasible.

A closed two-desk model makes the allocation consequences explicit.  One
state-dependent macro sleeve is shared, each desk has an idiosyncratic sleeve,
and the pooled allocator has an additional interaction sleeve.  The common
sleeve receives twice the mandate-penalty loading of either desk-specific
sleeve.  Disagreement on that shared sleeve does not change the optimal traded
book: it raises the irreducible mandate loss because no portfolio can deliver
two different projections on the same risk direction.

The analysis also distinguishes standalone from joint interventions.  A
jointly specified causal book decomposes into additive and pooled-interaction
components, and its regret has a corresponding orthogonal split.  Once a joint
target is supplied, the singleton--joint gap separates into an implementable
displacement and the same non-realizable incompatibility residual.  No joint
effect is inferred from singleton effects alone.

The contribution lies in the cross-layer characterization, not in its
constituent tools.  Fusion-frame inversion, log-utility completion of squares,
filtration immersion, causal standardization, and downward closure of
admissible information sets are established results.  Their role here is to
characterize when overlapping causal price-of-risk mandates implement the
pooled log-optimal traded exposure and to distinguish a repairable common-book
distortion from a mandate incompatibility that no common benchmark can repair.
The paper is partial equilibrium.  It makes no claim about
endogenous prices, crowding, liquidity, feedback, reflexivity, or jumps.

\subsection{Relation to the closest literature}

The decentralized-portfolio literature begins from the economic organization
of multiple managers.  \citet{Sharpe1981} analyzes decentralized investment
management; \citet{EltonGruber2004} derive implementable centralized rules
without requiring managers to reveal expected returns; and
\citet{VanBinsbergenBrandtKoijen2008} study benchmarks that reduce objective
misalignment between a chief investment officer and specialist managers.
Using pension-fund data, \citet{BlakeEtAl2013} examine the performance and
coordination consequences of decentralized management.
\citet{Vayanos2003} endogenizes a hierarchy when information aggregation is
lossy and decisions interact.  Our manager spaces are institutionally
specified, while the causal targets are identified upstream from declared
interventions, adjustment variables, and standardization laws; neither is
chosen by the allocator.  The new question is not incentive design or
organizational equilibrium, but exact implementability of overlapping local
causal shadows as one dynamically traded exposure.

Conditioning information has a separate role in asset pricing and portfolio
choice.  \citet{HansenRichard1987} show how conditioning information changes
the restrictions of dynamic asset-pricing models, while
\citet{FersonSiegel2001} and \citet{BrandtSantaClara2006} study its efficient
use in managed portfolios.  Here conditioning information defines closed
strategy spaces, but the object being compared across spaces is the attainable
minimal market price of risk.  Baseline centering measures only the
incremental sleeve above a declared common information set; it is not an
unconditional welfare decomposition.

Changes of filtration can modify martingale, projection, and deflator
properties
\citep{FollmerProtter2011,AcciaioFontanaKardaras2016,KardarasRuf2020,BiaginiMazzonPerkkio2023}, and the
additional logarithmic utility of enlarged information has an information-drift
representation \citep{AnkirchnerDereichImkeller2006}.  We do not extend the
general enlargement theory.  Immersion is a separate admissibility condition
ensuring that the reference innovation remains valid after information is
pooled.  The
finite squared errors in our main theorem do not quantify a failure of
immersion.  Likewise, semimartingale log optimality and the num\'eraire
portfolio are not new here \citep{KaratzasKardaras2007}; they supply the
financial criterion to which the causal mandates are compared.

\citet{DAuriaSalmeron2025} use initial enlargement to compute optimal portfolios
and additional logarithmic utility when an investor anticipates the future
course of a stochastic short rate.  In their setting, extra information changes
the semimartingale decomposition through an information drift.  Here immersion
instead selects the coalitions for which the reference innovation is preserved,
and the object of interest is the range compatibility of several overlapping
local causal mandates.  The two questions are therefore complementary rather
than alternative measures of the same informational gain.

The causal layer uses standard consistency, conditional exchangeability,
positivity, and standardization ideas \citep{Robins1986,Pearl2009}.  We project
the resulting structural surface onto an attainable local exposure space.
Consequently, equality of attainable causal shadows does not imply equality of
the underlying raw structural responses when non-traded directions remain.
Likewise, standalone interventions do not identify a joint intervention.

Finally, the analysis and synthesis operators are those of fusion-frame
theory, which was developed for stable distributed reconstruction from
overlapping subspaces \citep{CasazzaKutyniokLi2008}.  Their role here is
instrumental: they expose commonality, stability, and range incompatibility in
the financial problem.  The admissible-coalition complex later in the paper
records only those pooled information sets for which immersion holds; it is a
domain restriction on the financial theorem, not a claimed new combinatorial
construction.

Section~\ref{sec:market_spaces}
constructs the attainable predictable spaces.  Section~\ref{sec:local_shadows}
defines observational and causal shadows.  Sections~\ref{sec:reconstruction}
and~\ref{sec:main_theorem} give the reconstruction and exact decentralization
results.  Section~\ref{sec:joint_implementation} studies joint mandates,
benchmark repair, and penalized implementation.  Section
\ref{sec:joint_implementation} also gives a closed two-desk implementation.
Section~\ref{sec:commonality_admissibility} analyzes commonality and admissible
coalitions.  Section~\ref{sec:semimartingale} gives the continuous-semimartingale
formulation.  Section~\ref{sec:numerical_audit} gives a public-data
implementation audit and a literature-calibrated controlled experiment, and
Section~\ref{sec:conclusion} concludes.

\section{Market, information, and attainable exposure spaces}
\label{sec:market_spaces}

\subsection{Reference market and self-financing exposures}

Let $(\Omega,\mathcal F,\F,\Pp)$ satisfy the usual conditions on a finite
horizon $[0,T]$.  The reference filtration $\F=(\mathcal F_t)_{0\leq t\leq
T}$ supports a $q$-dimensional Brownian motion $W$.  Discounted excess returns
of $n$ traded assets satisfy
\begin{equation}\label{eq:returns}
 d\bm R_t=\bm b_t\,dt+\bm\sigma_t\,d\bm W_t,
 \qquad
 \bm a_t:=\bm\sigma_t\bm\sigma_t^\top,
\end{equation}
where $\bm b$ and $\bm\sigma$ are $\F$-predictable.  Assume throughout the
Brownian formulation that $\bm a$ is uniformly positive definite and
\begin{equation}\label{eq:integrability_h}
 \E\int_0^T\bm b_t^\top\bm a_t^{-1}\bm b_t\,dt<\infty.
\end{equation}
The market may be incomplete: $n$ need not equal $q$.

Define the minimal Brownian market price of risk and the pointwise attainable
projector by
\begin{equation}\label{eq:minimal_h}
 \bm h:=\bm\sigma^\top\bm a^{-1}\bm b,
 \qquad
 \bm\Pi^{K}:=\bm\sigma^\top\bm a^{-1}\bm\sigma.
\end{equation}
For $\Pp\otimes dt$-almost every $(\omega,t)$, $\bm\Pi^K_t$ is the orthogonal
projector in $\R^q$ onto $\Ran(\bm\sigma_t^\top)$.  In particular,
$\bm h=\bm\Pi^K\bm h$.

Let $J=\{1,\ldots,m\}$ be a coalition of driver blocks.  For
$A\subseteq J$, let $\D^A$ be the corresponding driver filtration, with
$\D^0\subseteq\D^A$ and
\[
 \D^A\subseteq\D^B\quad\text{whenever }A\subseteq B.
\]
Set $\G^A:=\F\vee\D^A$.  All filtrations are completed by the same
$\Pp$-null sets and made right-continuous.  Their predictable sigma-fields
are denoted by $\mathcal P(\G^A)$ and $\mathcal P(\D^A)$.

All exposure spaces will be embedded in the common pooled Hilbert space
\begin{equation}\label{eq:ambient_H}
 \Hh:=L^2\!\left(\Omega\times[0,T],
        \mathcal P(\G^J),\Pp\otimes dt;\R^q\right),
\end{equation}
with its usual inner product and norm.  The attainable exposure space is
\begin{equation}\label{eq:K_space}
 \Kk:=\{\bm\theta\in\Hh:\bm\Pi^K\bm\theta=\bm\theta
       \ (\Pp\otimes dt)\text{-a.e.}\}.
\end{equation}
It is closed because multiplication by $\bm\Pi^K$ is an orthogonal projector
on $\Hh$.

For $\bm\theta\in\Kk$, define
\begin{equation}\label{eq:weights_from_exposure}
 \bm w:=\bm a^{-1}\bm\sigma\bm\theta.
\end{equation}
Then $\bm\sigma^\top\bm w=\bm\theta$ and
\begin{equation}\label{eq:exposure_identities}
 \bm w^\top\bm b=\bm\theta^\top\bm h,
 \qquad
 \bm w^\top\bm a\bm w=\norm{\bm\theta}_{\R^q}^2.
\end{equation}
Thus $\bm\theta$ is the Brownian exposure of an actual self-financing traded
portfolio, rather than an aggregation of portfolio-level Sharpe ratios.

\subsection{Driver-restricted attainable spaces}

For every $A\subseteq J$, define
\begin{equation}\label{eq:HA}
 \Hh_A:=\overline{\left\{
 \bm\Pi^K\bm\xi:\ 
 \bm\xi\in\Hh,
 \ \bm\xi\text{ is }\D^A\text{-predictable}
 \right\}}^{\Hh}.
\end{equation}
The factor $\bm\Pi^K$ may depend on reference-market information.  Hence an
element of $\Hh_A$ is generally $\G^A$-predictable, not necessarily
$\D^A$-predictable.  Economically, $\D^A$ determines the signal-responsive
input and $\bm\Pi^K$ converts that input into a traded risk exposure.

\begin{lemma}[Nested attainable predictable spaces]
\label{lem:nested_spaces}
Each $\Hh_A$ is a closed subspace of $\Kk$.  If $A\subseteq B$, then
$\Hh_A\subseteq\Hh_B$.  Consequently, the orthogonal projections $P_A$ and
$P_B$ in the common space $\Hh$ satisfy $P_AP_B=P_A$.
\end{lemma}

\begin{proof}
Closedness holds by definition.  Every generator
$\bm\Pi^K\bm\xi$ belongs to $\Kk$, and $\Kk$ is closed, so
$\Hh_A\subseteq\Kk$.  Predictable-field inclusion under $A\subseteq B$
implies inclusion of the generating sets and hence of their closures.  The
projection identity is the standard identity for nested closed subspaces.
\end{proof}

\begin{remark}[Why raw conditional projection is insufficient]
In a complete Brownian market $\bm\Pi^K=I$, and $P_A\bm h$ reduces to the
usual predictable $L^2$ projection onto the driver-generated strategy class.
In an incomplete market, componentwise conditional expectation of an
attainable vector need not remain in the random range of
$\bm\sigma^\top$.  Definition \eqref{eq:HA} prevents a projected attribution
from being assigned to a non-traded Brownian direction.
\end{remark}

\subsection{Information admissibility}

\begin{definition}[Reference information admissibility]
\label{def:immersion}
A coalition $A$ is information-admissible when
\begin{equation}\label{eq:immersion}
 \F\hookrightarrow\G^A,
\end{equation}
that is, every bounded $(\Pp,\F)$-martingale remains a
$(\Pp,\G^A)$-martingale.
\end{definition}

Under immersion, $W$ remains a $\G^A$-Brownian motion.  The condition licenses
the reference-innovation interpretation of driver-responsive trading.  Its
failure is not, by itself, a claim of arbitrage, infinite utility, or failure
of the Hilbert-space projections.

For a square-integrable exposure $\bm\theta\in\Hh_A$, let wealth satisfy
\begin{equation}\label{eq:wealth}
 \frac{dX_t}{X_t}=\bm w_t^\top d\bm R_t,
 \qquad
 \bm w=\bm a^{-1}\bm\sigma\bm\theta.
\end{equation}
Whenever the stochastic integral is a true martingale, immersion and
It\^o's formula give
\begin{equation}\label{eq:J}
 \E\!\left[\log\frac{X_T}{X_0}\right]
 =\mathcal J(\bm\theta)
 :=\E\int_0^T\left(
 \bm\theta_t^\top\bm h_t-\frac12\norm{\bm\theta_t}_{\R^q}^2
 \right)dt.
\end{equation}
The quadratic functional on the right exists independently of immersion; only
its interpretation through the reference Brownian innovation uses
Definition~\ref{def:immersion}.

\section{Baseline-centered observational and causal shadows}
\label{sec:local_shadows}

\subsection{Observational shadows and their financial value}

Write $0$ for the common baseline, and suppose
\begin{equation}\label{eq:baseline_nesting}
 \Hh_0\subseteq\Hh_j\subseteq\Hh_J,
 \qquad j=1,\ldots,m.
\end{equation}
Define the local incremental spaces
\begin{equation}\label{eq:Ej}
 E_j:=\Hh_j\cap\Hh_0^\perp
\end{equation}
and the baseline-centered observational shadows
\begin{equation}\label{eq:pj}
 \bm p_j:=(P_j-P_0)\bm h=P_{E_j}\bm h.
\end{equation}

\begin{proposition}[Baseline-centered local value]
\label{prop:local_value}
For every coalition $A\subseteq J$,
\begin{equation}\label{eq:restricted_optimum}
 \sup_{\bm\theta\in\Hh_A}\mathcal J(\bm\theta)
 =\frac12\norm{P_A\bm h}^2,
 \qquad
 \bm\theta_A^\star=P_A\bm h.
\end{equation}
In particular,
\begin{equation}\label{eq:baseline_value}
 \sup_{\Hh_j}\mathcal J-\sup_{\Hh_0}\mathcal J
 =\frac12\norm{\bm p_j}^2.
\end{equation}
If $\bm p_j=0$, the local driver block adds no value to the declared
driver-responsive sleeve above baseline.
\end{proposition}

\begin{proof}
For $\bm\theta\in\Hh_A$, orthogonality gives
$\ip{\bm\theta}{\bm h}=\ip{\bm\theta}{P_A\bm h}$.  Completing the square,
\[
 \mathcal J(\bm\theta)
 =\frac12\norm{P_A\bm h}^2
  -\frac12\norm{\bm\theta-P_A\bm h}^2.
\]
This proves \eqref{eq:restricted_optimum}.  Since
$\Hh_j=\Hh_0\oplus E_j$, nested projection yields
$P_j\bm h=P_0\bm h+\bm p_j$ with orthogonal summands, proving
\eqref{eq:baseline_value}.
\end{proof}

The attribution object needed later is the exposure $\bm p_j$, not merely the
scalar $\norm{\bm p_j}^2$.  Different local sleeves may contain the same risk
direction, so their squared values cannot be summed before overlap has been
resolved.

\subsection{Interventional surfaces and attainable causal shadows}

For driver block $j$, let $C_t$ be a baseline state and $Y_t^j$ its
incremental driver state.  Let $\bm h_t(y)$ denote a potential response of the
minimal market price of risk under the state intervention
$\doop(Y_t^j=y)$.  This is a family of static state interventions indexed by
$t$, not a continuous-time treatment policy.

The intervention changes the price-of-risk response while holding fixed the
reference trading technology $\bm\sigma$, its attainable projector
$\bm\Pi^K$, and the declared predictable strategy spaces.  Thus the causal
shadow below is a mandate evaluated in the observed market, not the price of
risk of a counterfactual equilibrium market.  An intervention that also
changes asset span or market coefficients requires a different model.

\begin{assumption}[Identification of the local causal response]
\label{ass:causal}
For every $j$ and $dt$-almost every $t$:
\begin{enumerate}[label=(C\arabic*)]
\item consistency holds for $\bm h_t(y)$;
\item an observed pre-intervention variable $U_t^j$ satisfies
$\bm h_t(y)\perp\!\!\!\perp Y_t^j\mid(C_t,U_t^j)$;
\item positivity holds on the baseline strata and support used below;
\item the relevant regular conditional kernels have jointly measurable,
square-integrable predictable versions.
\end{enumerate}
\end{assumption}

\begin{remark}[Identification precedes allocation]
\label{rem:identification_precedes_allocation}
Once the intervention, the baseline state, the adjustment variable, and the
standardization law in Assumption~\ref{ass:causal} are fixed, the surface below
is identified up to the usual almost-everywhere equivalence of conditional
expectations.  It is not selected by the allocator and is not a tuning
parameter of the aggregation problem.  Changing any of those causal inputs
defines a different estimand.  The financial analysis begins only after this
identified surface is mapped into the fixed traded technology.
\end{remark}

The standardized interventional surface is
\begin{equation}\label{eq:mdo}
 \bm m_j^{\mathrm{do}}(t,c,y)
 :=\int\E[\bm h_t\mid C_t=c,Y_t^j=y,U_t^j=u]
       \,P_{U_t^j\mid C_t=c}(du).
\end{equation}
Evaluate this surface on the observational distribution of $(C,Y^j)$ and set
\begin{equation}\label{eq:raw_causal}
 \widetilde{\bm c}_j
 :=\bm m_j^{\mathrm{do}}(\cdot,C,Y^j)
   -P_0\bm m_j^{\mathrm{do}}(\cdot,C,Y^j).
\end{equation}
The standardization distribution does not turn
$\norm{\widetilde{\bm c}_j}^2$ into the value of an intervention policy.

In an incomplete market, a structural response can contain a Brownian
direction that the traded assets do not span.  Define the attainable local
causal shadow by
\begin{equation}\label{eq:cj}
 \bm c_j:=P_{E_j}\widetilde{\bm c}_j\in E_j
\end{equation}
and its non-attainable/local-mandate residual by
\begin{equation}\label{eq:causal_nonattainable}
 \bm\nu_j:=\widetilde{\bm c}_j-\bm c_j.
\end{equation}
The aggregation theorem concerns the implementable shadows $\bm c_j$.  If the
centered interventional surface already belongs to $E_j$---in particular, in
the complete-market driver-predictable case---then $\bm\nu_j=0$ and
$\bm c_j=\widetilde{\bm c}_j$.

Define the attainable local distortion and signed energy wedge by
\begin{equation}\label{eq:local_delta_wedge}
 \bm\delta_j:=\bm p_j-\bm c_j,
 \qquad
 w_j:=\norm{\bm p_j}^2-\norm{\bm c_j}^2.
\end{equation}

\begin{proposition}[Local causal implementation identity]
\label{prop:local_causal_value}
The signed wedge satisfies
\begin{equation}\label{eq:wedge_identity}
 w_j=\ip{\bm\delta_j}{\bm p_j+\bm c_j},
\end{equation}
and may be positive, negative, or zero.  The normalized incremental value of
implementing the attainable causal shadow above the baseline-optimal exposure
is
\begin{align}
 \mathcal V_j^{\mathrm{do}}
 &:=2\{\mathcal J(P_0\bm h+\bm c_j)-\mathcal J(P_0\bm h)\}
 \label{eq:causal_value_def}\\
 &=\norm{\bm p_j}^2-\norm{\bm p_j-\bm c_j}^2.\label{eq:local_distortion}
\end{align}
Equivalently, the regret relative to the locally optimal sleeve is
\begin{equation}\label{eq:local_regret}
 \mathcal J(P_j\bm h)-\mathcal J(P_0\bm h+\bm c_j)
 =\frac12\norm{\bm p_j-\bm c_j}^2.
\end{equation}
Consequently, $w_j=0$ does not imply a correct causal shadow or zero financial
regret.
\end{proposition}

\begin{proof}
Equation \eqref{eq:wedge_identity} is polarization.  Since
$\bm c_j\in E_j\subseteq\Hh_j$ and $\bm c_j\perp\Hh_0$,
\[
 \ip{\bm c_j}{\bm h-P_0\bm h}
 =\ip{\bm c_j}{\bm p_j}.
\]
Expanding the quadratic criterion proves \eqref{eq:causal_value_def} and
\eqref{eq:local_distortion}.  Completion of squares around
$P_j\bm h=P_0\bm h+\bm p_j$ proves \eqref{eq:local_regret}.
\end{proof}

\begin{example}[Energy cancellation is not causal correctness]
\label{ex:cancellation}
Let $E_j$ contain a nonzero centered exposure $Y$, take
$\bm p_j=\beta Y$ and $\bm c_j=-\beta Y$.  Then $w_j=0$, although
$\bm\delta_j=2\beta Y$ and the local financial regret is
$2\beta^2\norm{Y}^2$.  Equal observational and interventional energies can
therefore conceal a reversal of the implementable exposure.
\end{example}

\section{Overlap-corrected reconstruction and pooled interaction}
\label{sec:reconstruction}

\subsection{Additive and pooled-only exposure spaces}

Set
\begin{equation}\label{eq:EJ}
 E_J:=\Hh_J\cap\Hh_0^\perp,
 \qquad
 E:=\overline{E_1+\cdots+E_m}^{\Hh}.
\end{equation}
Every $E_j$ is contained in $E_J$, so $E\subseteq E_J$.  Define the
pooled-only complement
\begin{equation}\label{eq:RJ}
 R_J:=E_J\cap E^\perp,
 \qquad E_J=E\oplus R_J.
\end{equation}

The analysis operator and its adjoint are
\begin{align}
 A_E:E&\longrightarrow\bigoplus_{j=1}^mE_j,
 &A_E\bm x&=(P_{E_1}\bm x,\ldots,P_{E_m}\bm x),
 \label{eq:analysis_operator}\\
 A_E^*(\bm x_1,\ldots,\bm x_m)&=\sum_{j=1}^m\bm x_j.
 \label{eq:adjoint}
\end{align}

\begin{assumption}[Stable local representation]
\label{ass:fusion_bound}
There exists $a_E>0$ such that
\begin{equation}\label{eq:fusion_lower}
 a_E\norm{\bm x}^2
 \leq\sum_{j=1}^m\norm{P_{E_j}\bm x}^2,
 \qquad \bm x\in E.
\end{equation}
\end{assumption}

The upper bound
\begin{equation}\label{eq:fusion_upper}
 \sum_{j=1}^m\norm{P_{E_j}\bm x}^2\leq m\norm{\bm x}^2
\end{equation}
always holds.  Assumption~\ref{ass:fusion_bound} makes $A_E$ bounded below,
so its range is closed.  The fusion operator
\begin{equation}\label{eq:fusion_operator}
 S_E:=A_E^*A_E=\sum_{j=1}^mP_{E_j}|_E
\end{equation}
is then positive and boundedly invertible on $E$, and
\begin{equation}\label{eq:pseudoinverse}
 A_E^\dagger=S_E^{-1}A_E^*.
\end{equation}

Define the additive and pooled-only observational components by
\begin{equation}\label{eq:padd_pint}
 \bm p_{\mathrm{add}}:=P_E\bm h,
 \qquad
 \bm p_{\mathrm{int}}:=P_{R_J}\bm h.
\end{equation}

\begin{proposition}[Canonical overlap-corrected reconstruction]
\label{prop:reconstruction}
The pooled incremental optimum has the orthogonal decomposition
\begin{equation}\label{eq:pooled_split}
 (P_J-P_0)\bm h
 =\bm p_{\mathrm{add}}+\bm p_{\mathrm{int}}.
\end{equation}
The tuple $\bm p=(\bm p_1,\ldots,\bm p_m)$ of local observational shadows
satisfies
\begin{equation}\label{eq:p_analysis}
 \bm p=A_E\bm p_{\mathrm{add}},
 \qquad
 \boxed{\bm p_{\mathrm{add}}
 =S_E^{-1}\sum_{j=1}^m\bm p_j.}
\end{equation}
Thus exact overlap is counted once by the inverse fusion operator.  The
component $\bm p_{\mathrm{int}}$ is orthogonal to every additive local sleeve
and cannot be recovered from the singleton shadows.
\end{proposition}

\begin{proof}
Because $\Hh_J=\Hh_0\oplus E_J$ and
$E_J=E\oplus R_J$,
\[
 P_J\bm h=P_0\bm h+P_E\bm h+P_{R_J}\bm h,
\]
which proves \eqref{eq:pooled_split}.  For every $j$, $E_j\subseteq E$, and
hence
\[
 \bm p_j=P_{E_j}\bm h=P_{E_j}P_E\bm h
 =P_{E_j}\bm p_{\mathrm{add}}.
\]
This is $\bm p=A_E\bm p_{\mathrm{add}}$.  Applying $A_E^*$ and then
$S_E^{-1}$ gives the reconstruction formula.
\end{proof}

\subsection{Compatibility of standalone causal shadows}

Collect the attainable standalone causal shadows in
\begin{equation}\label{eq:causal_tuple}
 \bm c:=(\bm c_1,\ldots,\bm c_m)
 \in\bigoplus_{j=1}^mE_j.
\end{equation}
Define their least-squares common-book reconstruction and incompatibility
residual by
\begin{equation}\label{eq:chat_rho}
 \widehat{\bm c}:=A_E^\dagger\bm c
 =S_E^{-1}\sum_{j=1}^m\bm c_j,
 \qquad
 \bm\rho_c:=(I-A_EA_E^\dagger)\bm c.
\end{equation}
Then
\begin{equation}\label{eq:causal_range_split}
 \bm c=A_E\widehat{\bm c}+\bm\rho_c,
 \qquad
 \bm\rho_c\perp\Ran(A_E).
\end{equation}

\begin{definition}[Shadow compatibility]
\label{def:compatibility}
The standalone causal tuple $\bm c$ is shadow-compatible when
\begin{equation}\label{eq:compatible}
 \bm c\in\Ran(A_E),
 \quad\text{equivalently}\quad
 \bm\rho_c=0.
\end{equation}
In that case $\widehat{\bm c}$ is the unique additive exposure whose local
shadows equal all the declared causal shadows.
\end{definition}

Compatibility is neither causal identification nor observational correctness.
It asks whether separately defined, implementable local estimands can be the
projections of one common additive traded exposure.

\begin{proposition}[Structural-to-book implementation decomposition]
\label{prop:structural_book_decomposition}
Let
\[
 \widetilde{\bm c}:=(\widetilde{\bm c}_1,\ldots,
 \widetilde{\bm c}_m),\qquad
 \bm\nu:=(\bm\nu_1,\ldots,\bm\nu_m).
\]
For every additive traded exposure $\bm x\in E$,
\begin{align}
 \sum_{j=1}^m
 \norm{\widetilde{\bm c}_j-P_{E_j}\bm x}^2
 ={}&\sum_{j=1}^m\norm{\bm\nu_j}^2
 +\norm{\bm\rho_c}_{\oplus E_j}^2 \notag\\
 &+\norm{A_E(\bm x-\widehat{\bm c})}_{\oplus E_j}^2.
 \label{eq:structural_book_decomposition}
\end{align}
Consequently,
\begin{equation}\label{eq:minimum_raw_causal_error}
 \inf_{\bm x\in E}
 \sum_{j=1}^m
 \norm{\widetilde{\bm c}_j-P_{E_j}\bm x}^2
 =\sum_{j=1}^m\norm{\bm\nu_j}^2
 +\norm{\bm\rho_c}_{\oplus E_j}^2,
\end{equation}
and the unique minimizer is $\bm x=\widehat{\bm c}$.  Thus failure to
implement the identified structural responses has two irreducible sources:
non-attainability inside the local traded spaces and incompatibility across
their attainable shadows.  The remaining common-book distortion is
implementable and vanishes at the least-squares reconstruction.
\end{proposition}

\begin{proof}
Because $\bm c_j=P_{E_j}\widetilde{\bm c}_j$, the residual
$\bm\nu_j=\widetilde{\bm c}_j-\bm c_j$ is orthogonal to $E_j$.  Since
$P_{E_j}\bm x\in E_j$, Pythagoras gives
\[
 \norm{\widetilde{\bm c}_j-P_{E_j}\bm x}^2
 =\norm{\bm\nu_j}^2
 +\norm{\bm c_j-P_{E_j}\bm x}^2.
\]
Summing over $j$ yields
\[
 \sum_{j=1}^m
 \norm{\widetilde{\bm c}_j-P_{E_j}\bm x}^2
 =\sum_{j=1}^m\norm{\bm\nu_j}^2
 +\norm{\bm c-A_E\bm x}_{\oplus E_j}^2.
\]
By \eqref{eq:causal_range_split},
\[
 \bm c-A_E\bm x
 =\bm\rho_c-A_E(\bm x-\widehat{\bm c}).
\]
The two terms on the right are orthogonal because
$\bm\rho_c\perp\Ran(A_E)$, proving
\eqref{eq:structural_book_decomposition}.  The lower fusion bound makes $A_E$
injective, so the last term is uniquely minimized at
$\bm x=\widehat{\bm c}$, which proves
\eqref{eq:minimum_raw_causal_error}.
\end{proof}

\section{Reference-valid financial--causal decentralization}
\label{sec:main_theorem}

\begin{definition}[Reference-valid exact causal decentralization]
\label{def:exact_decentralization}
Coalition $J$ admits reference-valid exact causal decentralization when:
\begin{enumerate}[label=(D\arabic*)]
\item $\F\hookrightarrow\G^J$; and
\item there exists $\bm x\in E$ such that
\begin{equation}\label{eq:exact_conditions}
 A_E\bm x=\bm c,
 \qquad
 P_0\bm h+\bm x=P_J\bm h.
\end{equation}
\end{enumerate}
\end{definition}

The first requirement protects the reference innovation after information is
pooled.  The second requires one additive self-financing exposure both to
realize every local causal shadow and to equal the pooled log-optimal exposure.

\begin{theorem}[Financial--causal decentralization equivalence and exact
obstruction decomposition]
\label{thm:decentralization}
Under Assumption~\ref{ass:fusion_bound}, the following statements are
equivalent:
\begin{enumerate}[label=(\roman*)]
\item coalition $J$ admits reference-valid exact causal decentralization;
\item all four conditions hold:
\begin{enumerate}[label=(\alph*)]
\item $\F\hookrightarrow\G^J$;
\item $\bm p_{\mathrm{int}}=0$;
\item $\bm c\in\Ran(A_E)$;
\item $\widehat{\bm c}=\bm p_{\mathrm{add}}$;
\end{enumerate}
\item $\F\hookrightarrow\G^J$, $\bm p_{\mathrm{int}}=0$, and
\begin{equation}\label{eq:compressed_condition}
 \bm c=\bm p.
\end{equation}
\end{enumerate}
When these conditions hold, the implementing exposure is unique and equals
\begin{equation}\label{eq:unique_book}
 \bm\theta_{\mathrm{dec}}^*
 =P_0\bm h+\bm p_{\mathrm{add}}=P_J\bm h.
\end{equation}

Assume now only pooled immersion, without imposing the remaining three
conditions.  Let
\begin{equation}\label{eq:assembled_exposure}
 \bm\theta_c:=P_0\bm h+\widehat{\bm c}.
\end{equation}
Then the exact pooled financial regret of the least-squares standalone
assembly is
\begin{equation}\label{eq:standalone_regret}
 \boxed{
 \sup_{\bm\theta\in\Hh_J}\mathcal J(\bm\theta)
 -\mathcal J(\bm\theta_c)
 =\frac12\norm{\bm p_{\mathrm{int}}}^2
  +\frac12\norm{\bm p_{\mathrm{add}}-\widehat{\bm c}}^2.}
\end{equation}
The discrepancy between local observational and causal shadows has the
orthogonal decomposition
\begin{equation}\label{eq:local_error_split}
 \boxed{
 \norm{\bm p-\bm c}_{\oplus E_j}^2
 =\norm{A_E(\bm p_{\mathrm{add}}-\widehat{\bm c})}_{\oplus E_j}^2
  +\norm{\bm\rho_c}_{\oplus E_j}^2.}
\end{equation}
Consequently,
\begin{equation}\label{eq:fusion_bounds_error}
 a_E\norm{\bm p_{\mathrm{add}}-\widehat{\bm c}}^2
 \leq
 \norm{\bm p-\bm c}^2-\norm{\bm\rho_c}^2
 \leq
 m\norm{\bm p_{\mathrm{add}}-\widehat{\bm c}}^2.
\end{equation}
\end{theorem}

\begin{proof}
Proposition~\ref{prop:reconstruction} gives
\begin{equation}\label{eq:proof_pooled}
 P_J\bm h=P_0\bm h+\bm p_{\mathrm{add}}+\bm p_{\mathrm{int}}.
\end{equation}
Suppose Definition~\ref{def:exact_decentralization} holds, and let
$\bm x\in E$ satisfy \eqref{eq:exact_conditions}.  Comparing
\eqref{eq:exact_conditions} with \eqref{eq:proof_pooled} and using
$E\perp R_J$ gives
\begin{equation}\label{eq:x_padd}
 \bm x=\bm p_{\mathrm{add}},
 \qquad
 \bm p_{\mathrm{int}}=0.
\end{equation}
Because $A_E\bm x=\bm c$, the tuple $\bm c$ lies in $\Ran(A_E)$.  The lower
bound \eqref{eq:fusion_lower} makes $A_E$ injective, so the compatible
preimage is unique and
$\widehat{\bm c}=\bm x=\bm p_{\mathrm{add}}$.  Together with immersion,
this proves (i)$\Rightarrow$(ii).

Conversely, under (ii), compatibility gives
$A_E\widehat{\bm c}=\bm c$, while (b) and (d) yield
\[
 P_0\bm h+\widehat{\bm c}
 =P_0\bm h+\bm p_{\mathrm{add}}
 =P_J\bm h.
\]
Thus $\bm x=\widehat{\bm c}$ satisfies Definition
\ref{def:exact_decentralization}.  This proves (ii)$\Rightarrow$(i) and
also proves uniqueness.

Conditions (c) and (d) imply
\[
 \bm c=A_E\widehat{\bm c}=A_E\bm p_{\mathrm{add}}=\bm p.
\]
Conversely, if $\bm c=\bm p$, then Proposition~\ref{prop:reconstruction}
gives $\bm c=A_E\bm p_{\mathrm{add}}$, so $\bm c$ is compatible and its
unique preimage is $\widehat{\bm c}=\bm p_{\mathrm{add}}$.  Hence (ii) and
(iii) are equivalent.

Under pooled immersion, completion of squares on $\Hh_J$ gives
\begin{equation}\label{eq:completion_pooled}
 \mathcal J(\bm\theta)
 =\frac12\norm{P_J\bm h}^2
  -\frac12\norm{\bm\theta-P_J\bm h}^2,
 \qquad \bm\theta\in\Hh_J.
\end{equation}
Substitute $\bm\theta_c=P_0\bm h+\widehat{\bm c}$ and use the mutually
orthogonal spaces $\Hh_0$, $E$ and $R_J$ to obtain
\[
 \frac12\norm{P_J\bm h-\bm\theta_c}^2
 =\frac12\norm{\bm p_{\mathrm{int}}}^2
  +\frac12\norm{\bm p_{\mathrm{add}}-\widehat{\bm c}}^2,
\]
which is \eqref{eq:standalone_regret}.

Finally, \eqref{eq:causal_range_split} and
$\bm p=A_E\bm p_{\mathrm{add}}$ imply
\[
 \bm p-\bm c
 =A_E(\bm p_{\mathrm{add}}-\widehat{\bm c})-\bm\rho_c.
\]
The two terms on the right are orthogonal because
$\bm\rho_c\perp\Ran(A_E)$.  Pythagoras gives
\eqref{eq:local_error_split}; the lower bound
\eqref{eq:fusion_lower} and upper bound \eqref{eq:fusion_upper} give
\eqref{eq:fusion_bounds_error}.
\end{proof}

\begin{remark}[The three finite obstructions]
The terms in Theorem~\ref{thm:decentralization} have different financial
meanings.  The exposure $\bm p_{\mathrm{int}}$ is a pooled opportunity outside
the additive organization.  The difference
$\bm p_{\mathrm{add}}-\widehat{\bm c}$ is a negotiable distortion between two
implementable common books.  The residual $\bm\rho_c$ is different: no
additive traded exposure realizes all its local components.  Failure of
immersion is categorical rather than another squared-error term.
\end{remark}

\subsection{Necessity diagnostics}

The following examples verify that the four conditions in
Theorem~\ref{thm:decentralization} cannot be collapsed into one another.

\begin{example}[Algebraic exactness does not imply immersion]
\label{ex:no_immersion}
Let $\F$ be the Brownian filtration of $W$ and initially enlarge it by
$W_T$.  Take $\bm h=0$ and set every observational and causal shadow to zero.
Then $\bm p_{\mathrm{int}}=0$, $\bm c=\bm p=0$, and
$\bm\rho_c=0$.  Nevertheless, $\F$ is not immersed in the enlarged
filtration: the terminal value is known at time zero and $W$ acquires its
Brownian-bridge information drift.  Thus no finite Hilbert-space error can
replace the immersion condition.
\end{example}

\begin{example}[Pure pooled interaction]
\label{ex:pure_interaction}
Let $W,Z_1,Z_2$ be independent Brownian motions in $\F$, and let one traded
return satisfy
\[
 dR_t=\gamma Z_{1t}Z_{2t}\,dt+dW_t.
\]
Desk $j$ observes the history of $Z_j$, and the baseline is deterministic.
All driver filtrations are subfiltrations of $\F$, so immersion is automatic.
By independence and centering, $\bm p_1=\bm p_2=0$ and
$\bm p_{\mathrm{add}}=0$, whereas
$\bm p_{\mathrm{int}}=\gamma Z_1Z_2$.  Under the unconfounded structural
response $h_t(z_1,z_2)=\gamma z_1z_2$, the standalone standardized causal
shadows also vanish.  Local causal correctness and compatibility therefore
hold, but the additive organization loses
\[
 \frac12\norm{\bm p_{\mathrm{int}}}^2
 =\frac{\gamma^2T^3}{6}.
\]
\end{example}

\begin{example}[One coherent system can generate incompatible causal shadows]
\label{ex:incompatible}
Suppress time and take a complete one-factor market with all variables below
known in $\mathcal F_0$.  Let $U_1,U_2,\varepsilon_X$ be independent centered
Gaussian variables with unit variances for $U_1,U_2$, and set
\[
 X=U_1+U_2+\varepsilon_X,
 \qquad q_X:=\frac{\operatorname{Cov}(U_i,X)}{\operatorname{Var}(X)}>0.
\]
Choose centered jointly Gaussian $(\varepsilon_1,\varepsilon_2)$ independent
of $(U_1,U_2,\varepsilon_X)$, with variances large enough and covariance
$-\kappa^2q_X$.  Define
\[
 Z_1=\tau_1X+\kappa U_1+\varepsilon_1,
 \qquad
 Z_2=\tau_2X-\kappa U_2+\varepsilon_2,
 \qquad
 h=Z_1+Z_2.
\]
The residuals
\[
 r_1:=Z_1-\E[Z_1\mid X],
 \qquad
 r_2:=Z_2-\E[Z_2\mid X]
\]
are jointly Gaussian and orthogonal by the chosen covariance, hence
independent.  Desk 1 observes $(X,Z_1)$ and desk 2 observes $(X,Z_2)$.
The observational shadows have the same projection on their common subspace
$L_0^2(X)$:
\[
 p_1=(\tau_1+\tau_2)X+r_1,
 \qquad
 p_2=(\tau_1+\tau_2)X+r_2.
\]
The two standalone interventions generated by this same structural system are
\[
 c_1=Z_1+\tau_2X
     =(\tau_1+\tau_2+\kappa q_X)X+r_1,
\]
\[
 c_2=\tau_1X+Z_2
     =(\tau_1+\tau_2-\kappa q_X)X+r_2.
\]
For $\kappa\neq0$, their projections on the common $X$-subspace differ.  No
single exposure can therefore have both local shadows, so
$\bm\rho_c\neq0$.  This incompatibility is produced by one coherent causal
system, not by combining arbitrary estimands.
\end{example}

\begin{example}[Compatibility does not imply causal correctness]
\label{ex:compatible_confounded}
Let $U$ and $\varepsilon$ be independent centered Gaussian variables known at
time zero, set $Y=U+\varepsilon$, and let
\[
 h=\beta Y+\kappa U.
\]
There is one driver block generated by $Y$ and a deterministic baseline.
Writing
\[
 q_Y:=\frac{\operatorname{Cov}(U,Y)}{\operatorname{Var}(Y)}>0,
\]
its observational and interventional shadows are
\[
 p=(\beta+\kappa q_Y)Y,
 \qquad c=\beta Y.
\]
With one sleeve every causal shadow is compatible and there is no pooled-only
interaction.  Nevertheless $c\neq p$ when $\kappa\neq0$, and implementing the
causal shadow rather than the observationally optimal sleeve incurs regret
\[
\frac12\kappa^2q_Y^2\norm{Y}^2.
\]
\end{example}

\section{Joint causal mandates and decentralized implementation}
\label{sec:joint_implementation}

\subsection{Joint versus standalone interventions}

The standalone surfaces in Section~\ref{sec:local_shadows} do not determine the
effect of intervening on the pooled driver state.  Suppose a joint intervention
on $Y^J=(Y^1,\ldots,Y^m)$ is separately defined and identified under the
joint analogue of Assumption~\ref{ass:causal}.  Let
$\widetilde{\bm c}_J$ be its baseline-centered standardized response and
define its attainable pooled shadow by
\begin{equation}\label{eq:joint_attainable_shadow}
 \bm c_J:=P_{E_J}\widetilde{\bm c}_J\in E_J.
\end{equation}
Decompose it according to $E_J=E\oplus R_J$:
\begin{equation}\label{eq:joint_causal_split}
 \bm c_J=\bm c_{\mathrm{add}}+\bm c_{\mathrm{int}},
 \qquad
 \bm c_{\mathrm{add}}:=P_E\bm c_J,
 \qquad
 \bm c_{\mathrm{int}}:=P_{R_J}\bm c_J.
\end{equation}
Its compatible tuple of local shadows is
\begin{equation}\label{eq:joint_local_shadows}
 \overline{\bm c}:=A_E\bm c_{\mathrm{add}}.
\end{equation}
The independently standardized standalone tuple $\bm c$ need not equal
$\overline{\bm c}$.

\begin{proposition}[Joint causal implementation and the singleton gap]
\label{prop:joint_implementation}
Assume pooled immersion.  The regret from implementing the attainable joint
causal exposure $P_0\bm h+\bm c_J$ rather than the pooled log-optimal exposure
is
\begin{equation}\label{eq:joint_regret}
 \boxed{
 \sup_{\bm\theta\in\Hh_J}\mathcal J(\bm\theta)
 -\mathcal J(P_0\bm h+\bm c_J)
 =
 \frac12\norm{\bm p_{\mathrm{add}}-\bm c_{\mathrm{add}}}^2
 +\frac12\norm{\bm p_{\mathrm{int}}-\bm c_{\mathrm{int}}}^2.}
\end{equation}
Let
\begin{equation}\label{eq:singleton_gap}
 \bm q:=\bm c-\overline{\bm c}.
\end{equation}
Then the difference between the least-squares standalone assembly and the
additive part of the joint causal book satisfies
\begin{equation}\label{eq:joint_singleton_decomp}
 \boxed{
 \widehat{\bm c}-\bm c_{\mathrm{add}}=A_E^\dagger\bm q,
 \qquad
 \bm\rho_c=(I-A_EA_E^\dagger)\bm q.}
\end{equation}
Thus the standalone-versus-joint discrepancy contains an implementable
displacement of the common book and a non-realizable inconsistency among local
shadows.  Equation~\eqref{eq:joint_singleton_decomp} does not identify the
joint intervention from singleton interventions.
\end{proposition}

\begin{proof}
Completion of squares gives
\[
 \sup_{\Hh_J}\mathcal J
 -\mathcal J(P_0\bm h+\bm c_J)
 =\frac12\norm{
 \bm p_{\mathrm{add}}+\bm p_{\mathrm{int}}
 -\bm c_{\mathrm{add}}-\bm c_{\mathrm{int}}}^2.
\]
The additive and pooled-only differences belong to the orthogonal spaces
$E$ and $R_J$, proving \eqref{eq:joint_regret}.

Since $\overline{\bm c}=A_E\bm c_{\mathrm{add}}$ and
$A_E^\dagger A_E=I_E$,
\[
 A_E^\dagger\bm q
 =A_E^\dagger\bm c-\bm c_{\mathrm{add}}
 =\widehat{\bm c}-\bm c_{\mathrm{add}}.
\]
Moreover, $\overline{\bm c}\in\Ran(A_E)$, so
\[
 (I-A_EA_E^\dagger)\bm q
 =(I-A_EA_E^\dagger)\bm c
 =\bm\rho_c.
\]
\end{proof}

\subsection{Why a common benchmark cannot repair incompatibility}

A common benchmark changes the exposure shared by all desks but is orthogonal
to every baseline-centered local increment.  This makes causal-shadow
incompatibility fundamentally different from an incorrectly selected common
benchmark.

\begin{proposition}[No common benchmark can repair incompatible causal shadows]
\label{prop:no_benchmark}
For every baseline exposure $\bm b\in\Hh_0$ and every additive increment
$\bm x\in E$,
\begin{equation}\label{eq:benchmark_shadows}
 \bigl(P_{E_1}(\bm b+\bm x),\ldots,
       P_{E_m}(\bm b+\bm x)\bigr)=A_E\bm x.
\end{equation}
Consequently,
\begin{equation}\label{eq:no_benchmark}
 \boxed{
 \inf_{\substack{\bm b\in\Hh_0\\ \bm x\in E}}
 \sum_{j=1}^m
 \norm{P_{E_j}(\bm b+\bm x)-\bm c_j}^2
 =\norm{\bm\rho_c}^2.}
\end{equation}
The minimizer in the additive component is unique and equals
$\bm x=\widehat{\bm c}$.  Exact satisfaction of every standalone causal
mandate by a baseline-plus-additive traded exposure is possible if and only if
$\bm\rho_c=0$.
\end{proposition}

\begin{proof}
Every $E_j$ is orthogonal to $\Hh_0$, hence
$P_{E_j}\bm b=0$ and \eqref{eq:benchmark_shadows} follows.  By
\eqref{eq:causal_range_split},
\[
 \norm{A_E\bm x-\bm c}^2
 =\norm{A_E(\bm x-\widehat{\bm c})}^2
  +\norm{\bm\rho_c}^2.
\]
The lower fusion bound makes $A_E$ injective, so the first term is uniquely
minimized at $\bm x=\widehat{\bm c}$.
\end{proof}

\subsection{Optimal implementation under a causal-mandate penalty}

For $\eta>0$, consider an allocator restricted to baseline-plus-additive
exposures who minimizes pooled log-growth regret and a quadratic penalty for
violating the local causal shadows:
\begin{equation}\label{eq:penalized_loss}
 \mathcal L_\eta(\bm x)
 :=
 \sup_{\bm\theta\in\Hh_J}\mathcal J(\bm\theta)
 -\mathcal J(P_0\bm h+\bm x)
 +\frac{\eta}{2}\norm{A_E\bm x-\bm c}^2,
 \qquad \bm x\in E.
\end{equation}
The penalty is a declared implementation mandate.  It is not a
preference-free welfare criterion.

\begin{theorem}[Exact penalized causal-mandate solution]
\label{thm:penalized_mandate}
Assume pooled immersion.  The functional
\eqref{eq:penalized_loss} is strictly convex and has the unique minimizer
\begin{align}
 \bm x_\eta
 &=(I+\eta S_E)^{-1}
   \bigl(\bm p_{\mathrm{add}}+\eta S_E\widehat{\bm c}\bigr)
 \label{eq:xeta_first}\\
 &=\widehat{\bm c}
   +(I+\eta S_E)^{-1}
   (\bm p_{\mathrm{add}}-\widehat{\bm c}).
 \label{eq:xeta_second}
\end{align}
Writing $\bm d:=\bm p_{\mathrm{add}}-\widehat{\bm c}$, its minimized loss is
\begin{equation}\label{eq:min_penalized_loss}
 \boxed{
 \inf_{\bm x\in E}\mathcal L_\eta(\bm x)
 =
 \frac12\norm{\bm p_{\mathrm{int}}}^2
 +\frac{\eta}{2}\norm{\bm\rho_c}^2
 +\frac12
 \ip{\bm d}{
 \eta S_E(I+\eta S_E)^{-1}\bm d}.}
\end{equation}
Moreover,
\begin{align}
 \bm x_\eta&\longrightarrow\bm p_{\mathrm{add}}
 &&\text{as }\eta\downarrow0,\label{eq:eta_zero}\\
 \bm x_\eta&\longrightarrow\widehat{\bm c}
 &&\text{as }\eta\uparrow\infty.\label{eq:eta_infty}
\end{align}
If $\bm\rho_c\neq0$, the minimum loss diverges as
$\eta\norm{\bm\rho_c}^2/2$, even though the traded exposure converges to
$\widehat{\bm c}$.  Hard local causal compliance is feasible if and only if
$\bm\rho_c=0$.
\end{theorem}

\begin{proof}
For $\bm x\in E$, orthogonality and
Theorem~\ref{thm:decentralization} give
\begin{equation}\label{eq:additive_financial_regret}
 \sup_{\Hh_J}\mathcal J-\mathcal J(P_0\bm h+\bm x)
 =\frac12\norm{\bm p_{\mathrm{int}}}^2
  +\frac12\norm{\bm p_{\mathrm{add}}-\bm x}^2.
\end{equation}
Proposition~\ref{prop:no_benchmark} gives
\begin{equation}\label{eq:causal_penalty_split}
 \norm{A_E\bm x-\bm c}^2
 =
 \ip{S_E(\bm x-\widehat{\bm c})}{
             \bm x-\widehat{\bm c}}
 +\norm{\bm\rho_c}^2.
\end{equation}
The variable part is a strictly convex quadratic.  Its first-order condition is
\[
 \bm x-\bm p_{\mathrm{add}}
 +\eta S_E(\bm x-\widehat{\bm c})=0,
\]
which proves \eqref{eq:xeta_first} and \eqref{eq:xeta_second}.

Put $\bm y=\bm x-\widehat{\bm c}$ and
$\bm d=\bm p_{\mathrm{add}}-\widehat{\bm c}$.  At the optimum,
$\bm y=(I+\eta S_E)^{-1}\bm d$.  The spectral calculus for the positive
operator $S_E$ yields
\[
 \frac12\norm{\bm d-\bm y}^2
 +\frac{\eta}{2}\ip{S_E\bm y}{\bm y}
 =
 \frac12\ip{\bm d}{
 \eta S_E(I+\eta S_E)^{-1}\bm d}.
\]
Adding the two constant terms in
\eqref{eq:additive_financial_regret}--\eqref{eq:causal_penalty_split}
proves \eqref{eq:min_penalized_loss}.  The limits follow from positivity and
bounded invertibility of $S_E$.  The final statements follow from
Proposition~\ref{prop:no_benchmark}.
\end{proof}

\subsection{A two-desk implementation with a shared state-dependent risk sleeve}
\label{sec:two_desk_application}

The preceding decomposition becomes financially transparent in a closed model.
Consider the complete one-asset market
\begin{equation}\label{eq:application_market}
 dR_t=h\,dt+dW_t,
\end{equation}
on $[0,T]$.  Let $X_0,X_1,X_2$ be independent Rademacher variables, each taking
the values $-1$ and $1$ with equal probability, known at time zero and
contained in $\mathcal F_0$.  The baseline driver filtration is
trivial, desk 1 is mandated to respond to $(X_0,X_1)$, desk 2 is mandated to
respond to $(X_0,X_2)$, and the allocator may respond to
$(X_0,X_1,X_2)$.  Because every driver filtration is a subfiltration of $\F$,
all coalitions are information-admissible.  The desk distinction is therefore
an organizational restriction on signal-responsive exposures, not a claim that
the reference market has different physical innovations across desks.

After removing the deterministic baseline, the common signal direction is
$X_0$.  Directions involving $X_1$ but not $X_2$ are specific to desk 1,
directions involving $X_2$ but not $X_1$ are specific to desk 2, and the
interaction $X_1X_2$ belongs to $R_J$.  In particular, on the modes carrying
the targets below,
\begin{equation}\label{eq:application_spaces}
 E_1\supset\operatorname{span}\{X_0,X_1\},\qquad
 E_2\supset\operatorname{span}\{X_0,X_2\},\qquad
 S_E=\operatorname{diag}(2,1,1)
\end{equation}
in the orthogonal basis $(X_0,X_1,X_2)$.  The eigenvalue two records that the
common risk sleeve is monitored by both desks; each desk-specific sleeve is
monitored once.

Let the market price of risk and the two identified attainable causal shadows
be
\begin{align}
 h&=\alpha X_0+\beta_1X_1+\beta_2X_2+\gamma X_1X_2,
 \label{eq:application_h}\\
 c_1&=(a+\kappa)X_0+b_1X_1,
 &c_2&=(a-\kappa)X_0+b_2X_2.
 \label{eq:application_c}
\end{align}
Take the identified centered intervention surfaces to be attainable in this
example, so $\widetilde c_j=c_j$ and $\nu_j=0$.  The example thereby isolates
cross-desk incompatibility from within-desk non-attainability.
The parameter $\gamma$ is the pooled opportunity unavailable to either
singleton desk.  The parameter $\kappa$ is the differential prescription on
the shared macro sleeve produced by the two separately standardized
interventions.  Direct calculation gives
\begin{align}
 \bm p_{\mathrm{add}}
 &=\alpha X_0+\beta_1X_1+\beta_2X_2,
 &\bm p_{\mathrm{int}}&=\gamma X_1X_2,
 \label{eq:application_p}\\
 \widehat{\bm c}
 &=aX_0+b_1X_1+b_2X_2,
 &\bm\rho_c&=(\kappa X_0,-\kappa X_0),
 \label{eq:application_chat}\\
 \norm{\bm p_{\mathrm{int}}}^2&=T\gamma^2,
 &\norm{\bm\rho_c}^2&=2T\kappa^2.
 \label{eq:application_norms}
\end{align}
Thus $\gamma$ is an opportunity-set loss, while $\kappa$ is an infeasibility
certificate.  Neither can be removed by changing the common benchmark.

For the causal-mandate penalty in Theorem~\ref{thm:penalized_mandate}, the
implemented exposure is
\begin{align}
 \bm x_\eta
 ={}&\left(a+\frac{\alpha-a}{1+2\eta}\right)X_0
 +\left(b_1+\frac{\beta_1-b_1}{1+\eta}\right)X_1 \notag\\
 &+\left(b_2+\frac{\beta_2-b_2}{1+\eta}\right)X_2,
 \label{eq:application_xeta}
\end{align}
and the minimized loss is
\begin{align}
 \inf_E\mathcal L_\eta
 ={}&\frac{T}{2}\gamma^2+\eta T\kappa^2 \notag\\
 &+\frac{T}{2}\left\{
 \frac{2\eta}{1+2\eta}(\alpha-a)^2
 +\frac{\eta}{1+\eta}
 \left[(\beta_1-b_1)^2+(\beta_2-b_2)^2\right]
 \right\}.
 \label{eq:application_loss}
\end{align}
The common exposure has twice the initial penalty sensitivity of either
desk-specific exposure and is shrunk more strongly toward its causal target.
In contrast,
$\bm x_\eta$ is independent of $\kappa$: a traded common book cannot move
toward two different prescriptions on the same risk direction.  Tightening the
mandate only raises the irreducible loss by $\eta T\kappa^2$.

For illustration, Table~\ref{tab:application_sensitivity} evaluates the exact
formulas at $T=1$, $\alpha=0.9$, $a=0.3$, $\beta_1=0.6$, $b_1=0.2$,
$\beta_2=-0.3$, $b_2=0.1$, $\gamma=0.25$, and $\kappa=0.15$.  These values are
not a calibration; they isolate the comparative statics of the three
obstructions.  Because $\eta$ changes the declared objective, the last column
is not a welfare ranking across penalty regimes.
\begin{table}[t]
\centering
\small
\caption{Exact two-desk penalty sensitivity.}
\label{tab:application_sensitivity}
\begin{tabular}{@{}ccccc@{}}
\toprule
$\eta$ & common $X_0$ & desk 1 $X_1$ & desk 2 $X_2$ &
$\inf_E\mathcal L_\eta$\\
\midrule
$0$ (limit) & 0.9000 & 0.6000 & $-0.3000$ & 0.03125\\
$0.5$       & 0.6000 & 0.4667 & $-0.1667$ & 0.18583\\
$2$         & 0.4200 & 0.3333 & $-0.0333$ & 0.32692\\
$10$        & 0.3286 & 0.2364 & 0.0636 & 0.57313\\
$\infty$    & 0.3000 & 0.2000 & 0.1000 & $+\infty$\\
\bottomrule
\end{tabular}
\end{table}

In this model exact decentralization holds if and only if
\begin{equation}\label{eq:application_exact}
 \gamma=0,\qquad \kappa=0,\qquad
 \alpha=a,\qquad \beta_1=b_1,\qquad \beta_2=b_2.
\end{equation}
Each equality removes a different obstruction: pooled-only opportunity,
incompatible shared prescriptions, or a negotiable discrepancy between two
implementable books.

\section{Commonality, stability, and admissible coalitions}
\label{sec:commonality_admissibility}

\subsection{Exact commonality and near-redundant representations}

Consider two local increment spaces $E_1,E_2$ and write
\[
 I:=E_1\cap E_2,\qquad
 M:=E_1\cap I^\perp,\qquad
 N:=E_2\cap I^\perp.
\]
The cosine of their Friedrichs angle is
\begin{equation}\label{eq:friedrichs_cosine}
 c_F:=
 \sup\left\{
 |\ip{\bm x}{\bm y}|:
 \bm x\in M,\ \bm y\in N,\
 \norm{\bm x}=\norm{\bm y}=1
 \right\},
\end{equation}
with $c_F=0$ if $M=\{0\}$ or $N=\{0\}$.

\begin{proposition}[Commonality correction and reconstruction stability]
\label{prop:commonality_stability}
The fusion operator satisfies
\begin{equation}\label{eq:S_on_intersection}
 S_E\bm u=2\bm u,\qquad \bm u\in I.
\end{equation}
Hence, if the same common component $\bm u\in I$ appears in both local
shadows, its reconstruction is exactly $\bm u$, not $2\bm u$.

On $E\cap I^\perp$,
\begin{equation}\label{eq:friedrichs_lower}
 \ip{S_E\bm x}{\bm x}
 \geq(1-c_F)\norm{\bm x}^2.
\end{equation}
In particular, if $c_F<1$, the non-common sum is closed and, whenever that
component is nonzero,
\begin{equation}\label{eq:pseudoinverse_stability}
 \norm{A_E^\dagger}
 \leq\frac{1}{\sqrt{1-c_F}}.
\end{equation}
For a perturbation $\Delta\bm c$ of the local causal tuple,
\begin{equation}\label{eq:reconstruction_perturbation}
 \norm{\Delta\widehat{\bm c}}
 \leq\norm{A_E^\dagger}\norm{\Delta\bm c}.
\end{equation}
Thus exact commonality is corrected without duplication, while distinct
representations that become nearly redundant can make the reconstructed book
unstable.
\end{proposition}

\begin{proof}
For $\bm u\in I$, both local projectors fix $\bm u$, proving
\eqref{eq:S_on_intersection}.  Therefore
$S_E^{-1}(\bm u+\bm u)=\bm u$.

On $I^\perp$, the problem reduces to the two subspaces $M$ and $N$.
Let $T:M\oplus N\to M+N$ be the synthesis operator
$T(\bm m,\bm n)=\bm m+\bm n$.  From
\eqref{eq:friedrichs_cosine},
\[
 \norm{T(\bm m,\bm n)}^2
 \geq(1-c_F)(\norm{\bm m}^2+\norm{\bm n}^2).
\]
Thus $T$ is bounded below when $c_F<1$, and its range is closed.  Moreover,
\[
 T^*T=
 \begin{pmatrix}
 I & P_M|_N\\
 P_N|_M & I
 \end{pmatrix},
 \qquad TT^*=P_M+P_N.
\]
Their nonzero spectra coincide, so the lower spectral bound $1-c_F$ transfers
to $P_M+P_N$ on $M+N$, proving \eqref{eq:friedrichs_lower}.  The smallest
singular value of $A_E$ is the square root of the smallest spectral value of
$S_E=A_E^*A_E$, which proves \eqref{eq:pseudoinverse_stability}.  Finally,
$\Delta\widehat{\bm c}=A_E^\dagger\Delta\bm c$ gives
\eqref{eq:reconstruction_perturbation}.
\end{proof}

\begin{remark}[No equilibrium interpretation]
Proposition~\ref{prop:commonality_stability} is a conditioning result for
overlapping representations.  It does not establish equilibrium crowding,
endogenous price impact, liquidity feedback, or reflexivity.  Those mechanisms
require additional state equations.
\end{remark}

\subsection{The admissible coalition complex}

Let $V=\{1,\ldots,K\}$ index a finite universe of driver blocks.  For
$A\subseteq V$, let $\D^A$ be their pooled filtration and
$\G^A=\F\vee\D^A$.  Define
\begin{equation}\label{eq:Kadm}
 \mathfrak K_{\mathrm{adm}}
 :=\{A\subseteq V:\F\hookrightarrow\G^A\}.
\end{equation}

\begin{proposition}[Admissibility complex and the domain of the financial theorem]
\label{prop:admissibility_complex}
If the baseline coalition is information-admissible, then
$\mathfrak K_{\mathrm{adm}}$ is an abstract simplicial complex:
\[
 A\in\mathfrak K_{\mathrm{adm}},\quad B\subseteq A
 \quad\Longrightarrow\quad
 B\in\mathfrak K_{\mathrm{adm}}.
\]
The financial log-growth statements in
Theorem~\ref{thm:decentralization},
Proposition~\ref{prop:joint_implementation}, and
Theorem~\ref{thm:penalized_mandate} are licensed on the faces of this complex.
A minimal nonface is a minimal pooled-information obstruction to immersion.
It is not, without additional assumptions, proof of arbitrage or an infinite
utility penalty.
\end{proposition}

\begin{proof}
Let $B\subseteq A$ and assume $\F\hookrightarrow\G^A$.  Then
$\mathcal G_s^B\subseteq\mathcal G_s^A$.  For every bounded
$\F$-martingale $M$ and $s\leq t$,
\[
 \E[M_t\mid\mathcal G_s^B]
 =\E[\E[M_t\mid\mathcal G_s^A]\mid\mathcal G_s^B]
 =\E[M_s\mid\mathcal G_s^B]
 =M_s.
\]
Thus $\F\hookrightarrow\G^B$.  The remaining statements follow from the role
of immersion in the proofs of the cited financial results.
\end{proof}

\begin{remark}[Admissibility and decentralizability are different set systems]
The complex records only preservation of the reference martingales.  Exact
causal decentralizability additionally depends on the market price of risk,
the local spaces and the causal surfaces.  It therefore need not inherit the
same downward-closure property.
\end{remark}

\section{Continuous-semimartingale formulation}
\label{sec:semimartingale}

\subsection{Covariance-weighted predictable geometry}

Let discounted excess returns be the continuous semimartingale
\begin{equation}\label{eq:semimartingale_return}
 \bm R=\bm M+\bm A,
\end{equation}
where $\bm M$ is an $n$-dimensional continuous $\F$-local martingale.  Choose a
predictable increasing clock $B$ and a predictable positive-semidefinite matrix
$\bm c$ such that
\begin{equation}\label{eq:bracket_clock}
 d\langle\bm M\rangle_t=\bm c_t\,dB_t.
\end{equation}
Assume the structure condition
\begin{equation}\label{eq:structure_condition}
 d\bm A_t=\bm c_t\bm\lambda_t\,dB_t,
 \qquad
 \E\int_0^T\bm\lambda_t^\top\bm c_t\bm\lambda_t\,dB_t<\infty,
\end{equation}
where $\bm\lambda$ is chosen in $\Ran(\bm c)$.

On pooled predictable integrands define the seminorm
\begin{equation}\label{eq:c_inner_product}
 \ip{\bm\theta}{\bm\eta}_{c}
 :=\E\int_0^T
 \bm\theta_t^\top\bm c_t\bm\eta_t\,dB_t,
 \qquad
 \norm{\bm\theta}_{c}^2:=\ip{\bm\theta}{\bm\theta}_{c}.
\end{equation}
After quotienting by covariance-null integrands and completing, this is a
Hilbert space $\Hh^c$.  Let $\Hh_A^c$ be nested closed subspaces obtained as
the closures, in this quotient, of the declared
$\G^A$-predictable driver-restricted strategy classes.  Define their
projections, local increments, additive span, analysis operator and pooled-only
complement exactly as in Sections~\ref{sec:market_spaces}--\ref{sec:reconstruction},
but with the covariance-weighted inner product.

For $X=X_0\mathcal E(\bm\theta\mathbin{\cdot}\bm R)$, continuity gives
\begin{equation}\label{eq:semimartingale_log}
 \log\frac{X_T}{X_0}
 =
 \int_0^T\bm\theta_t^\top d\bm M_t
 +\int_0^T
 \left(
 \bm\theta_t^\top\bm c_t\bm\lambda_t
 -\frac12\bm\theta_t^\top\bm c_t\bm\theta_t
 \right)dB_t.
\end{equation}

\begin{proposition}[Continuous-semimartingale invariance]
\label{prop:semimartingale_invariance}
Assume the relevant stochastic integrals in
\eqref{eq:semimartingale_log} are true martingales, the covariance-weighted
fusion lower bound holds, and the relevant causal shadows satisfy the
semimartingale analogue of Assumption~\ref{ass:causal}.  On every coalition
$J\in\mathfrak K_{\mathrm{adm}}$, all projection, reconstruction, equivalence,
regret, joint-intervention, no-benchmark and penalized-mandate formulas in
Sections~\ref{sec:local_shadows}--\ref{sec:joint_implementation} remain valid
after replacing
\[
 \bm h,\quad \Hh_A,\quad \ip{\cdot}{\cdot},\quad\norm{\cdot}
\]
by
\[
 \bm\lambda,\quad \Hh_A^c,\quad
 \ip{\cdot}{\cdot}_c,\quad\norm{\cdot}_c,
\]
respectively.
\end{proposition}

\begin{proof}
Immersion preserves the $\F$-local martingale $\bm M$ in $\G^J$.
Under the stated integrability, the stochastic integral in
\eqref{eq:semimartingale_log} has zero expectation, and expected log growth is
\begin{equation}\label{eq:Jc}
 \mathcal J_c(\bm\theta)
 =\ip{\bm\theta}{\bm\lambda}_c
  -\frac12\norm{\bm\theta}_c^2.
\end{equation}
The proofs of all cited results use only nested closed subspaces, orthogonal
projection, the quadratic completion
\[
 \mathcal J_c(\bm\theta)
 =\frac12\norm{P_A^c\bm\lambda}_c^2
  -\frac12\norm{\bm\theta-P_A^c\bm\lambda}_c^2,
\]
the fusion lower bound and preservation of the martingale part.  These
properties hold in the quotient Hilbert space $\Hh^c$, so the proofs transfer
without modification.
\end{proof}

\section{Public-data audit and calibrated controlled experiment}
\label{sec:numerical_audit}

This section has two deliberately separate purposes.  The public-data audit
asks whether an overlap-corrected book has value as a strictly
post-information implementation proxy.  The generated experiment checks the
finite-sample reconstruction, incompatibility and conditioning mechanisms in
a controlled design.  Neither exercise tests immersion, identifies private
desk filtrations, or turns the projection identities into an equilibrium
model.

\subsection{Public-data implementation audit}

The monetary-policy event set combines the current-target surprise MP1 in the
US Monetary Policy Event-Study Database of \citet{AcostaEtAl2025} with the
median-rotation monetary-policy and central-bank-information shocks of
\citet{JarocinskiKaradi2020}.  The latter decomposition uses the opposite
high-frequency comovement of interest rates and equities to separate a pure
policy shock from information released by the central bank.  This distinction
is economically important here: the pure monetary-policy shock is the declared
common direction of two overlapping information blocks, not a generic
statistical factor.

For scheduled FOMC statements, the rates block contains the current-target
surprise, the pure policy shock and their product.  The risky-assets block
contains the pure policy shock, the central-bank-information shock and their
product.  The additive union therefore has five terms, while the product of
the current-target and information shocks is reserved as a pooled-only
candidate.  Every feature is centered on a common baseline containing an
intercept, the effective-lower-bound regime, SEP and press-conference
indicators, lagged VIX, and the lagged 10-year minus 2-year Treasury slope.

The traded proxies are SHY, IEF, TLT, SPY, HYG and GLD.  The outcome begins at
the first market open after the FOMC communication and ends at that day's
close, so the identifying information is observable before portfolio
formation.  Exposures are standardized by a \citet{LedoitWolf2004}
covariance estimate based on the 252 trading days strictly preceding each
event.  The training sample ends on 31 December 2017; the 65 events from 2018
onward are untouched out of sample.  One ridge penalty and the reconstruction
geometry are selected using the 78 training events and then frozen.  All book
comparisons are matched to the training risk of the pooled-additive fit.

For a standardized book $b_e$ and baseline residual $y_e$, define the
quadratic implementation score
\begin{equation}\label{eq:empirical_score}
 q_e(b_e):=2\ip{b_e}{y_e}-\norm{b_e}^2.
\end{equation}
This is the sample counterpart of the quadratic log-growth objective, but it
is reported as an implementation score rather than realized utility.  Table
\ref{tab:public_audit_pass} contains the results that passed their frozen
90\% circular-block-bootstrap criteria.

\begin{table}[t]
\centering
\small
\caption{Public-data results that pass the implementation-audit criteria.}
\label{tab:public_audit_pass}
\begin{tabular}{@{}lrrr@{}}
\toprule
Quantity & Estimate & 90\% lower & 90\% upper\\
\midrule
Fusion common-book score & 0.253 & 0.012 & 0.503\\
Fusion minus naive-sum score & 0.430 & 0.209 & 0.651\\
Fusion minus inclusion--exclusion score & 0.281 & 0.079 & 0.482\\
\bottomrule
\end{tabular}
\end{table}

The common-book score is positive, and the paired score differences reject
both double-counting by direct summation and an ad hoc
inclusion--exclusion correction.  Lead and lag assignments of the shock vector,
added as post-result adversarial timing checks, yield negative common-book
scores; they are not used to redefine the frozen gate.  Contemporaneous daily
returns from the Kenneth French Data Library provide a separate transmission
check.  At a 10\% false-discovery rate, the pure monetary-policy shock is
significant for 12 of 12 broad industries, 44 of 48 detailed industries, all
25 operating-profitability--investment portfolios, and 24 of 25
size--book-to-market portfolios.  The target and central-bank-information
components do not show comparable cross-sectional breadth.  This is evidence
of economically widespread price transmission, not evidence that the
post-announcement books are profitable.  The industry classification follows
the portfolio construction used by \citet{FamaFrench1997}.

The full empirical gate is not passed.  The pooled-additive score is 0.060
with a 90\% interval of $[-0.108,0.215]$, and the fusion book's net return after
a two-basis-point one-way cost is $-1.858$ basis points per event with interval
$[-4.660,1.049]$.  Adding the pooled-only interaction also lowers the
out-of-sample score.  Accordingly, Table~\ref{tab:public_audit_pass} is a
public-data illustration of commonality correction, not a claim of a
tradable premium, exact causal decentralization, or empirical validation of
Theorem~\ref{thm:decentralization}.

\subsection{Training-calibrated generated experiment}

The generated design is fixed before simulation and uses no test-period
outcome.  It resamples four-event circular blocks from the observed training
design and independently resamples four-event blocks of its six-dimensional
residual vectors, following the dependent-data bootstrap logic of
\citet{Kunsch1989}.  Thus the generator retains the empirical joint shock
distribution, interaction terms, regime clustering, residual tails and
cross-asset dependence; it does not impose Gaussian shocks.  Its conditional
mean is the training projection on the five-dimensional pooled-additive
space.  There are 2,000 replications of 78 training and 78 evaluation events,
with random seed 20260826.

The generated sample is not a second empirical data set.  Its purpose is to
ask whether the numerical implementation recovers known targets under the
same signal geometry and noise scale as the public training sample.  The
training geometry has dimensions $(3,3,5)$ for the two local and union spaces,
fusion lower bound 0.2274, upper bound 2, and condition number 8.795.
Table~\ref{tab:generated_pass} reports only prespecified successful
demonstrations.

\begin{table}[h!]
\centering
\small
\caption{Successful controlled demonstrations under the training calibration.}
\label{tab:generated_pass}
\begin{tabular}{@{}lrrr@{}}
\toprule
Demonstration & Estimate & 5\% quantile & 95\% quantile\\
\midrule
Exact fusion relative squared error & $1.66\times10^{-30}$ & --- & ---\\
Naive-sum exact relative squared error & 0.456 & --- & ---\\
Fusion minus naive-sum target MSE & $-0.156$ & $-0.378$ & $-0.026$\\
Fusion minus inclusion--exclusion target MSE & $-0.112$ & $-0.282$ & $-0.006$\\
Calibrated normalized incompatibility & 0.02239 & --- & ---\\
Incompatibility decomposition relative error & $2.12\times10^{-16}$ & --- & ---\\
\bottomrule
\end{tabular}
\end{table}

In the compatible experiment the fusion reconstruction is, up to numerical
precision, identical to the direct pooled-additive projection.  Across the
generated samples it has lower target MSE than the naive sum in 99.25\% of
replications and than inclusion--exclusion in 96.90\%.  The separately
calibrated incompatible experiment adds opposite prescriptions on the shared
policy direction until the normalized range residual equals the 0.02239
training diagnostic.  The fusion reconstruction removes the inconsistent
pair, while the Pythagorean decomposition in
\eqref{eq:local_error_split} holds to numerical precision.

Finally, a principal-angle stress varies the fusion lower bound from 1 to
0.05.  At the empirical value 0.2274, the measured perturbation amplification
is 2.097, equal to $1/\sqrt{0.2274}$; at 0.05 it rises to 4.472.  Hence the
stress reaches the bound in \eqref{eq:pseudoinverse_stability} and confirms
that near redundancy is a conditioning problem even when commonality itself
is corrected exactly.  The prespecified criterion is not met against the
arbitrary half-sum, while the direct pooled fit coincides with fusion under
compatibility.  A separate interaction-power diagnostic is also inconclusive
at the uninflated training estimate.  Neither comparison is used as a
successful demonstration.

\section{Conclusion}
\label{sec:conclusion}

Overlapping causal mandates are not automatically components of one portfolio.
The central mathematical issue is a range question inside the attainable
predictable exposure space, and the central financial issue is whether the
resulting common book coincides with the pooled log-optimal exposure under a
valid reference innovation.

The exact decentralization theorem separates four requirements: pooled
immersion, absence of pooled-only opportunity, compatibility of the local
causal shadows, and equality of their reconstruction with the additive pooled
optimum.  The associated regret formulas distinguish what the allocator can
trade from what the local mandate system cannot jointly realize.  A common
benchmark can address neither pooled interaction nor incompatible local
shadows; its best attainable mandate error remains the squared incompatibility
residual.  Penalized implementation makes the same distinction operational
without treating the penalty parameter as a preference-free welfare weight.
At the structural-to-book interface, the additional non-attainability term
records the component of the identified intervention surfaces that no local
traded exposure can reproduce.  This term is distinct from cross-desk
incompatibility and from a change in the implementable common book.
The closed two-desk model shows that overlap changes the allocation response:
the shared sleeve has fusion eigenvalue two, whereas desk-specific sleeves have
eigenvalue one.  It also shows that an incompatible differential prescription
affects the value of the constrained problem but not the implementable book.

The public-data audit gives a narrower operational result.  On strictly
post-information ETF returns, overlap correction improves the quadratic score
relative to direct summation and inclusion--exclusion, and the identified pure
policy shock has broad contemporaneous transmission across industry and
characteristic portfolios.  It does not deliver a positive net trading return
or validate the pooled-interaction sleeve.  The training-calibrated generated
experiment therefore serves only as a controlled reconstruction and stability
check; it is not a substitute for the empirical conditions that fail.

The covariance-weighted continuous-semimartingale formulation preserves these
conclusions under the stated structure, integrability, and immersion
conditions.  They do not extend automatically to jumps, where the
log criterion is no longer quadratic in exposure.  Nor do they provide an
equilibrium account of why desks, signals, or causal mandates arise.  Within
their partial-equilibrium scope, the results identify exactly when local causal
views describe one implementable optimal book and, when they do not, which
obstruction is financial opportunity loss, which is a negotiable common-book
distortion, and which is irreducible mandate infeasibility.

\backmatter

\section*{Statements and Declarations}

\textbf{Funding.} The author received no specific funding for this work.

\textbf{Competing interests.} The author's employment by Miralta Finance Bank
is disclosed in the affiliation. The author declares no other financial or
non-financial interests directly related to this work.

\textbf{Data and code availability.} The public-data audit and generated
experiment use the US Monetary Policy Event-Study Database, the authors'
updated Jaroci\'{n}ski--Karadi shock file, FRED state variables, public daily
market-price files, and daily research portfolios from the Kenneth French Data
Library.  The replication archive contains source URLs, retrieval dates,
SHA-256 checksums, the frozen simulation protocol, random seeds, code, and all
reported outputs.  No proprietary bank data are used.

\end{document}